\documentclass[submission,copyright,creativecommons]{eptcs}
\providecommand{\event}{Workshop on Horn Clauses for Verification and Synthesis (HVCS) 2014} % Name of the event you are submitting to
\usepackage{breakurl}             % Not needed if you use pdflatex only.
\usepackage{stmaryrd}
\usepackage{verbatim}
\usepackage{pdfpages}
\usepackage{caption}

\usepackage{algorithm}
\usepackage{algorithmicx}
\usepackage{algpseudocode}

\usepackage{amsfonts} 
\usepackage{amssymb}

\newcommand{\red}[1]{\textcolor{red}{#1}}
\newfloat{Example}{h}{Example}
\usepackage{multicol}% http://ctan.org/pkg/multicols
\usepackage{stfloats}
\usepackage[none]{hyphenat}
\usepackage{graphicx}
\usepackage{times}
\usepackage{amsmath}
\usepackage{microtype}
\usepackage{wrapfig}
\usepackage{hyperref}
\newtheorem{definition}{Definition}
\newtheorem{theorem}{Theorem}
\newtheorem{proof}{Proof}
\newtheorem{corollary}{Corollary}

\title{Runtime Verification Through Forward Chaining}
\author{Alan Perotti
\institute{University of Turin}
\email{perotti@di.unito.it}
\and
Guido Boella
\institute{University of Turin}
\email{boella@di.unito.it}
\and
Artur d'Avila Garcez
\institute{City University London}
\email{a.garcez@city.ac.uk}
}
\def\titlerunning{Runtime Verification Through Forward Chaining}
\def\authorrunning{A. Perotti, G. Boella \& A. d'Avila Garcez}
\begin{document}
\maketitle

\begin{abstract}
In this paper we present a novel rule-based approach for Runtime Verification of FLTL properties over finite but expanding traces. Our system exploits Horn clauses in implication form and relies on a forward chaining-based monitoring algorithm. This approach avoids the branching structure and exponential complexity typical of tableaux-based formulations, creating monitors with a single state and a fixed number of rules. This allows for a fast and scalable tool for Runtime Verification: we present the technical details together with a working implementation.\\
\end{abstract}

%%%%%%%%%%%%%%%%%%%%%%%%%%%%%%%%%%%%%%%%%%%%%%%%%%%%%%%%%%%%%%%%%%%%%%%%%%%%%%%%%%%%
%%%%%%%%%%%%%%%%%%%%%%%%%%%%%%%%%%%%%%%%%%%%%%%%%%%%%%%%%%%%%%%%%%%%%%%%%%%%%%%%%%%%
%%%%%%%%%%%%%%%%%%%%%%%%%%%%%%%%%%%%%%%%%%%%%%%%%%%%%%%%%%%%%%%%%%%%%%%%%%%%%%%%%%%%
\section{Introduction}
We are designing a framework for combining runtime verification and learning in connectionist models to improve the verification of compliance of systems based on business processes. By adapting formal specifications of such systems to include tolerable soft-violations occurring in real-practice to optimise the systems, we want to obtain a more realistic representation of compliance. Adaptation is the recent trend in Process Mining~\cite{AAlst}: the goal is “to discover, monitor and improve real processes (i.e., not assumed processes) by extracting knowledge from event logs readily available in today’s (information) systems”. Within this wider framework, this paper focuses on the introduction of a novel monitoring system, RuleRunner, built as a set of Horn clauses in implication form and exploiting forward chaining to perform runtime verification tasks. A RuleRunner system can be encoded in a recurrent neural network exploiting results from the Neural-Symbolic Integration~\cite{cilp} area, but this is outside the scope of this paper. \\
%An online prototype for RuleRunner can be found at {\href{http://www.di.unito.it/~perotti/RuleRunner.jnlp}{www.di.unito.it/$\sim$perotti/RuleRunner.jnlp}}.\\
This paper is structured as follows: Section 2 introduces background and related work, while Section 3 provides a technical introduction of our rule system. Section 4 provides experimental results and Section 5 ends the paper with final considerations and directions for future work.\\

\section{Background and Related Work}
\subsection{Horn Clauses and Chaining}
A Horn clause~\cite{horn} is a clause which contains at most one positive literal. The general format of such a clause is thus as follows:
$$\neg\alpha_1 \vee .. \vee \neg\alpha_n \vee \beta$$
This may be rewritten as an implication:
$$(\alpha_1 \wedge .. \wedge \alpha_n) \rightarrow \beta$$
where $\beta$ is called {\em head} and $(\alpha_1 \wedge .. \wedge \alpha_n)$ is called {\em body}. The two formulations are equivalent, and usually the former is called {\em disjunctive form} and the latter {\em implication form}. Horn clauses are used for knowledge representation and automatic reasoning; in particular, inference with Horn clauses can be done through backward or forward chaining. Backward chaining algorithms are goal-driven approaches that work their way from a given goal or query; it is implemented in logic programming (e.g. in Prolog) by SLD resolution~\cite{kow}. Forward chaining is a data-driven approach that starts with the available data and uses inference rules to extract more data until a goal is reached; it is a popular implementation strategy for production rule systems~\cite{hornprod}. \\
\subsection{Runtime Verification}
Runtime Verification (RV) relates an observed system with a formal property $\phi$ specifying some desired behaviour. An RV module, or monitor, is defined as a device that reads a trace and yields a certain verdict~\cite{rv}. A trace is a sequence of cells, which in turn are lists of observations occurring in a given discrete span of time. Runtime verification may work on finite (terminated), finite but continuously expanding, or on prefixes of infinite traces. While LTL is a standard semantic for infinite traces~\cite{ltl}, there are many semantics for finite traces: FLTL~\cite{fltl}, RVLTL~\cite{rvltl}, LTL3~\cite{ltl3}, LTL$\pm$~\cite{ltl+} just to name some. Since LTL semantics is based on infinite behaviours, the issue is to close the gap between properties specifying infinite behaviours and finite traces. In particular, FLTL differs from LTL as it offers two {\em next} operators ($X, \bar{X}$ in~\cite{rvltl}, $X,W$ in this paper), called respectively {\em strong} and {\em weak} next. Intuitively, the strong (and standard) $X$ operator is used to express with $X\phi$ that a next state must exist and that this next state has to satisfy property $\phi$. In contrast, the weak $W$ operator in $W\phi$ says that if there is a next state, then this next state has to satisfy the property $\phi$. More formally, let $u = a_0..a_{n-1}$ denote a finite trace of length $n$. The truth value of an FLTL formula $\psi$ (either $X\phi$ or $W\phi$) w.r.t. $u$ at position $i < n$, denoted by $[u,i\vDash \psi]$, is an element of $\mathbb{B}$ and is defined as follows:\\

\noindent
{\footnotesize
\begin{minipage}{0,45\textwidth}
\[
[u,i\vDash X\phi] = 
\begin{cases}
[u,i+1\vDash \phi],& \text{if i+1} < n\\
\perp,              & \text{otherwise}
\end{cases}
\]
\end{minipage}
\begin{minipage}{0,05\textwidth}
	\hspace*{-2mm}
\end{minipage}
\begin{minipage}{0,45\textwidth}
\[
[u,i\vDash W\phi] = 
\begin{cases}
[u,i+1\vDash \phi],& \text{if i+1} < n\\
\top,              & \text{otherwise}
\end{cases}
\]
\end{minipage}
}

\vspace*{5mm}
While RVLTL and LTL3 have been proven to hold interesting properties w.r.t. FLTL (see~\cite{rvltl}), we selected FLTL as we think it captures a more intuitive semantics when dealing with finite traces. Suppose to monitor $\phi = \Box a$ over a trace $t$, where $a$ is observed in all cells: we have that $[t\vDash \phi]$ equals, respectively, $\top$ in FLTL, $?$ in LTL3, and $T^p$ in RVLTL. If $t$ is seen as a prefix of a longer trace $t\sigma$, then LTL3 and RVLTL provide valuable information about how $\phi$ could be evaluated over $\sigma$. But if $t$ is a conclusive, self-contained trace (e.g. a daily set of transactions), then the FLTL semantics captures the intuitive positive answer to the query {\em does $a$ always hold in this trace?}\\
Several RV systems have been developed, and they can be clustered in three main approaches, based respectively on rewriting, automata and rules~\cite{rv}. Within rule based approaches, RuleR~\cite{barr} uses an original approach. It copes with the temporal dimension by introducing rules which may reactivate themselves in later stages of the reasoning, and RuleRunner is inspired by this powerful idea. However, RuleR rules may contain disjunctions in the head and therefore do not correspond to Horn clauses. Furthermore, RuleR creates alternative {\em observations expectations}, and therefore the application of forward-chaining inference mechanisms on a RuleR system creates a branching, Kripke-like {\em possible world structure}~\cite{kripke}. We focus on FLTL and encode each formula in a system of rules that correspond to Horn clauses and therefore allow to apply forward-chaining inference algorithms. The next section will describe the difference in the two approaches in more detail.

\section{The RuleRunner Rule System}
RuleRunner is a rule-based online monitor observing finite but expanding traces and returning an FLTL verdict. RuleRunner accepts formulae $\phi$ generated by the grammar:

\vspace*{-2mm}
$$\phi ::= true \mid \ a \mid \ !a \mid \phi \vee \phi \mid \phi \wedge \phi \mid \phi U \phi \mid X\phi \mid W\phi \mid \Diamond \phi \mid \Box \phi \mid END$$
\vspace*{-5mm}

\noindent
$a$ is treated as an atom and corresponds to a single observation in the trace. We assume, without loss of generality, that temporal formulae are in negation normal form (NNF), e.g., negation operators pushed inwards to propositional literals and cancellations applied. $W$ is the weak next operator. {\em END} is a special character that is added to the last cell of a trace to mark the end of the input stream.

\begin{algorithm}
	\caption{RuleRunner monitoring (abstract)}
	\label{alg:abs}
	{\scriptsize
\begin{algorithmic}[1]
	\Function{RR-monitoring}{$\phi$,trace $t$} 
	\State Build a monitor $RR_{\phi}$ encoding $\phi$
	\While{new cells exist in $t$}
	\State Observe the current cell
	\State Compute truth values of $\phi$ in the current cell of $t$
		\Comment{Evaluation rules}
	\If {$\phi$ is verified or falsified}
    \State \Return SUCCESS or FAILURE respectively
	\EndIf
	\State Set up the monitor for the next cell in $t$
		\Comment{Reactivation rules}
	\EndWhile
	\EndFunction
\end{algorithmic}
}
\end{algorithm}

Given an FLTL formula $\phi$ and a trace $t$, Algorithm \ref{alg:abs} provides an abstract description of the creation and runtime behaviour of a RuleRunner system monitoring $\phi$ over $t$. At first, a monitor encoding $\phi$ is computed. Second, the monitor enters the verification loop, composed by observing a new cell of the trace and computing the truth value of the property in the given cell. If the property is irrevocably satisfied or falsified in the current cell, RuleRunner outputs a binary verdict. If this is not the case (because the $\phi$ refers to cells ahead in the trace), the system shifts to the following cell and enters another monitoring iteration. The FLTL semantics guarantees that, if the trace ends, the verdict in the last cell of the trace is binary. RuleRunner is a runtime monitor, as it analyses one cell at a time and never needs to store past cells in memory nor peek into future ones.

%%%%%%%%%%%%%%%%%%%%%%%%%%%%%%%%%%%%%%%%%%%%%%%%%%%%%%%%%%%%%%%
%%%%%%%%%%%%%%%%%%%%%%%%%%%%%%%%%%%%%%%%%%%%%%%%%%%%%%%%%%%%%%%
%%%%%%%%%%%%%%%%%%%%%%%%%%%%%%%%%%%%%%%%%%%%%%%%%%%%%%%%%%%%%%%
\subsection{Building the rule system}

\begin{definition}
A RuleRunner system is a tuple $\langle R_E, R_R, S\rangle$, where $R_E$ ({\em evaluation rules}) and $R_R$ ({\em reactivation rules}) are rule sets, and $S$ (for {\em state}) is a set of active rules, observations and truth evaluations.
\end{definition}
Throughout this paper we mostly use the terms {\em State} and {\em rule}, as they are used in the Runtime Verification area. However, our rules correspond to Horn clauses in implication form, and what we call {\em State} corresponds to a {\em knowledge base}.\\

Given a finite set of observations $O$ and an FLTL formula $\phi$ over (a subset of) $O$, a state $S$ is a set of observations ($o \in O$), rule names ($R[\psi]$) and truth evaluations ($[\psi]V$); $V \in \{T,F,?\}$ is a truth value. A rule name $R[\psi]$ in $S$ means that the logical formula $\psi$ is under scrutiny, while a truth evaluation $[\psi]V$ means that the logical formula $\psi$ currently has the truth value $V$. The third truth value, $?$ ({\em undecided}), means that it is impossible to give a binary verdict in the current cell.\\

%The state evolves according to rules: RuleRunner is composed by evaluation and reactivation rules ($R_E,R_R$). 

Evaluation rules follow the pattern $R[\phi], [\psi^1]V, \dots, [\psi^n]V, \rightarrow [\phi]V$ and their role is to compute the truth value of a formula $\phi$ under verification, given the truth values of its direct subformulae $\psi^i$ (line 5 in Algorithm \ref{alg:abs}). For instance, $R[\Diamond \psi], [\psi]T \rightarrow [\Diamond \psi]T$ reads as {\em if $\Diamond \psi$ is being monitored and $\psi$ holds, then $\Diamond \psi$ is true}.\\

Reactivation rules follow the pattern $[\phi]? \rightarrow R[\phi], R[\psi^1], \dots, R[\psi^n]$ and the meaning is that if one formula is evaluated to undecided, that formula (together with its subformulae) is scheduled to be monitored again in the next cell of the trace (line 9 in Algorithm \ref{alg:abs}). For instance, $[\Diamond \psi]? \rightarrow R[\Diamond \psi], R[\psi]$ means that {\em if $\Diamond \psi$ was not irrevocably verified nor falsified in the current cell of the trace, both $\psi$ and $\Diamond \psi$ will be monitored again in the next cell}.\\

Evaluation rules are Horn clauses in implication form. Reactivation rules usually have several positive conjuncts in the head, and therefore a reactivation rule $A \rightarrow \beta_1,..\beta_n$ (where $A = \alpha_1,..,\alpha_m$) can be rewritten as $n$ separate Horn clauses $A \rightarrow \beta_1, .. , A \rightarrow \beta_n$. Having different rules with the same head is something to handle with care in case of backward chaining, as many inferential engines implement a depth-first search and therefore the order of these rules impacts on the result. This is not the case when applying forward chaining, as for all rules, if all the premises of the implication are known, then its
conclusion is added to the set of known facts.\\

A RuleRunner feature is that rules never involve disjunctions. In RuleR, for instance, the simple formula $\Diamond a$ is mapped to the rule $R_{\Diamond a} : \ \longrightarrow a \mid R_{\Diamond a}$ and its meaning, intuitively, is that, if $\Diamond a$ has to be verified, either $a$ is observed (thus satisfying the property) or the whole formula will be checked again (in the next cell of the trace). The same formula corresponds to the following set of rules in RuleRunner:\\

%\vspace*{-5mm}
\begin{minipage}[t]{0.5\textwidth}
$$R[\Diamond a], [a]T \rightarrow [\Diamond a]{T}$$
$$R[\Diamond a], [a]? \rightarrow [\Diamond a]{?}$$
$$R[\Diamond a], [a]F \rightarrow [\Diamond a]{?}$$
\end{minipage}
\begin{minipage}[t]{0.5\textwidth}
$$R[\Diamond a], [a]?, END \rightarrow [\Diamond a]{F}$$
$$[\Diamond a]? \rightarrow R[a], R[\Diamond a]$$
\end{minipage}
\vspace*{5mm}

The disjunction in the head of the RuleR rule corresponds to the additional constraints in the body of the RuleRunner rules. Therefore, where RuleR generates a set of alternative hypotheses and later matches them with actual observations, RuleRunner maintains a detailed state of exact information. This is achieved by means of evaluation tables: three-valued truth tables (as introduced by Lukasiewitz~\cite{three}) annotated with {\em qualifiers}. Each evaluation rule for $\phi$ corresponds to a single cell of the evaluation table for the main operator of $\phi$; a {\em qualifier} is a subscript letter providing additional information to $?$ truth values. Table \ref{fig:or} gives the example for disjunction. Qualifiers ($B,L,R$ in this case) are used to store and propagate detailed information about the verification status of formulae. \\

\begin{table}
\begin{center}
	\begin{minipage}[h]{.18\linewidth}
	%\vspace{0pt}
	\centering
	\begin{tabular}{ c|c|c|c| }
	\multicolumn{1}{c}{$\bigvee$}
	 &  \multicolumn{1}{c}{$T$}
	 &  \multicolumn{1}{c}{$?$}
	 & \multicolumn{1}{c}{$F$} \\
	\cline{2-4}
	$T$ & $T$ & $T$  & $T$ \\
	\cline{2-4}
	$?$ & $T$ & $?$  & $?$ \\
	\cline{2-4}
	$F$ & $T$ & $?$  & $F$ \\
	\cline{2-4}
	\end{tabular}
	%\vspace*{3mm}
	%\centering \caption{{\bf (B)}}
	\end{minipage}
	\begin{minipage}[h]{.07\linewidth}
	%\vspace{0pt}
	\centering
	\begin{tabular}{ c|}
	{}\\
	{}\\
	{}\\
	{}
	\end{tabular}
	%\vspace*{3mm}
	%\centering \caption{{\bf (B)}}
	\end{minipage}
	\begin{minipage}[t]{.2\linewidth}
	%\vspace{0pt}
	\centering
	\begin{tabular}{ c|c|c|c| }
	\multicolumn{1}{c}{$\bigvee_B$}
	 &  \multicolumn{1}{c}{$T$}
	 &  \multicolumn{1}{c}{$?$}
	 & \multicolumn{1}{c}{$F$} \\
	\cline{2-4}
	$T$ & $T$ & $T$  & $T$ \\
	\cline{2-4}
	$?$ & $T$ & $?_B$  & \red{$?_L$} \\
	\cline{2-4}
	$F$ & $T$ & $?_R$  & $F$ \\
	\cline{2-4}
	\end{tabular}
	\end{minipage}%
	\begin{minipage}[t]{.11\linewidth}
	%\vspace{0pt}
	\begin{tabular}{ c|c| }
	\multicolumn{1}{c}{}
	 & \multicolumn{1}{c}{$\bigvee_L$} \\
	\cline{2-2}
	$T$ & $T$  \\
	\cline{2-2}
	$?$ & $?_L$\\
	\cline{2-2}
	$F$ & $F$ \\
	\cline{2-2}
	\end{tabular}
	%\centering \caption{{\bf (B)}}
	\end{minipage}
	\begin{minipage}[t]{.11\linewidth}
	%\vspace{0pt}
	\begin{tabular}{ c|c| }
	\multicolumn{1}{c}{}
	 & \multicolumn{1}{c}{$\bigvee_R$} \\
	\cline{2-2}
	$T$ & $T$  \\
	\cline{2-2}
	$?$ & $?_R$\\
	\cline{2-2}
	$F$ & $F$ \\
	\cline{2-2}
	\end{tabular}
	\end{minipage}
\vspace*{2mm}
\caption{{\footnotesize truth table (left) and evaluation tables (right) for $\vee$}}
%\protect 
\label{fig:or}
\end{center}
\end{table}

For instance, if $\phi$ is undecided and $\psi$ is false when monitoring $\phi \vee \psi$ (highlighted cell in Table~\ref{fig:or}), $?_L$ means that the disjunction is undecided, but that its future verification state will depend on the truth value of the \red{L}eft disjunct. Note, in fact, how $\vee_L$ is a unary operator. An example for this is monitoring $\Diamond b \vee a$ against a cell including only $c$: $a$ is false, $\Diamond b$ is undecided (as $b$ may be observed in the future), and the whole disjunction will be verified/falsified in the following cells depending on $\Diamond b$ only.\\

\begin{figure}[ht!]
	\centering
	%\hspace*{-9mm}
	\includegraphics[scale=0.25]{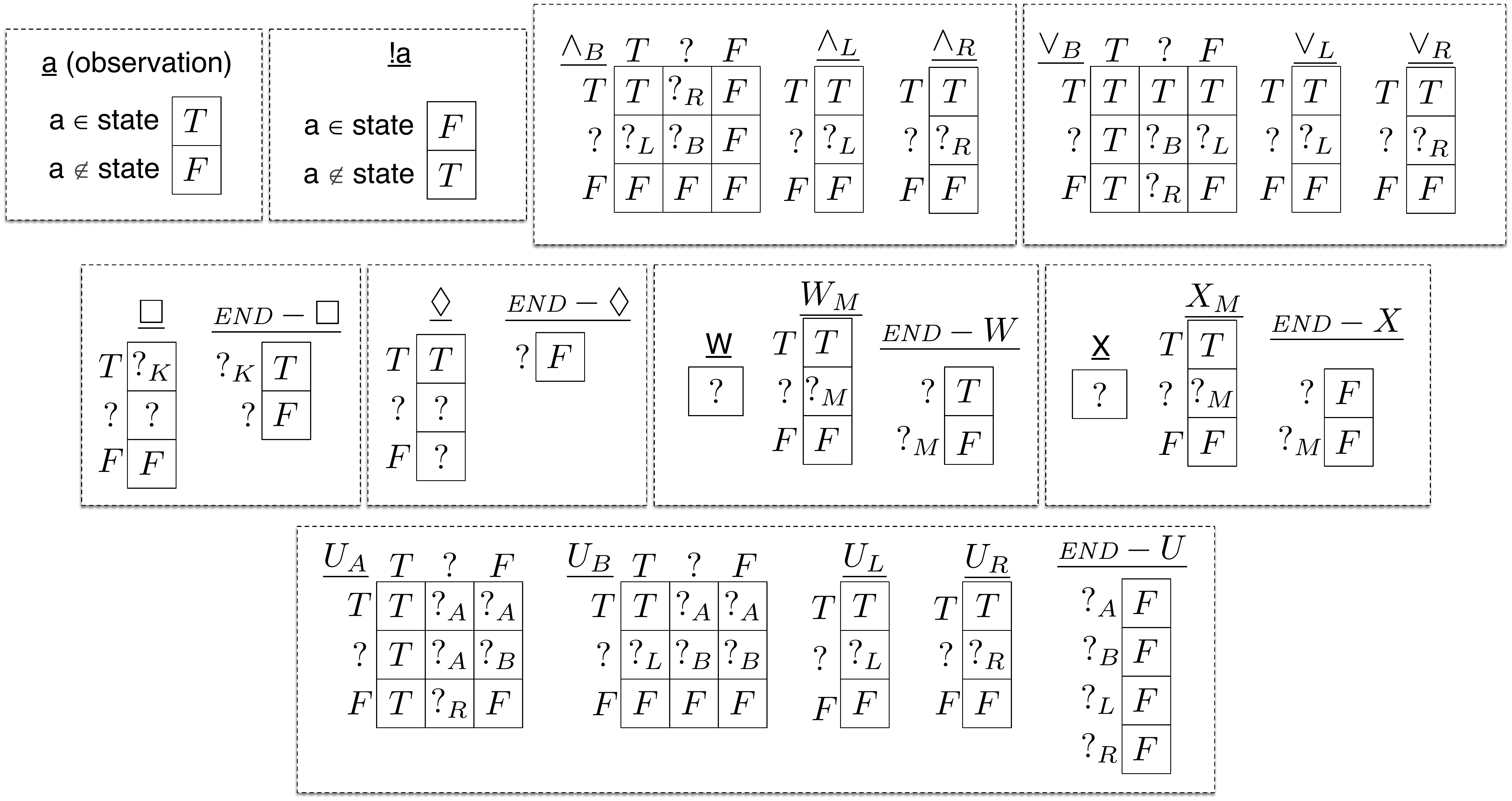}
	\vspace*{-3mm}
	\caption{Evaluation tables}
	\label{fig:tblz}
\end{figure}

\begin{algorithm}[ht!]
	\caption{Generation of rules}
	\label{alg:init}
	{\scriptsize
\begin{algorithmic}[1]
	\Function{Initialise}{$\phi$}
	\State $op \gets$ main operator
	\item[]
	\Comment{Apply recursively to subformula(e)}
	\If {$op \in \{\Box,\Diamond,X,W\}$}
	\State $\langle R_E^1, R_R^1, S^1 \rangle \gets$ Initialise($\psi^1$)
	\State $R_E \gets R_E^1$; 
	\State $R_R \gets R_R^1$; 
	\ElsIf {$op \in \{\vee,\wedge,U\}$}
	\State $\langle R_E^1, R_R^1, S^1 \rangle \gets$ Initialise($\psi^1$)
	\State $\langle R_E^2, R_R^2, S^2 \rangle \gets$ Initialise($\psi^2$)
	\State $R_E \gets R_E^1 \cup R_E^2$; $R_R \gets R_R^1 \cup R_R^2$; 
	\Else 
	\State $R_E \gets \emptyset$; $R_R \gets \emptyset$;
	\EndIf
	\item[]
	\Comment{Compute and add evaluation rules for main operator}
	\State $Cells \gets$op's-evaluation-tables
	\ForAll{cell $\in$ Cells}
	\State Convert cell to single rule $r_e$, substituting formula names
	\State $R_E \gets R_E \cup r_e$
	\EndFor
	\If {$\phi$-is-main-formula}
	\State $R_E \gets R_E \cup ([\phi]T \rightarrow SUCCESS)$
	\State $R_E \gets R_E \cup ([\phi]F \rightarrow FAILURE)$
	\State $R_E \gets R_E \cup ([\phi]? \rightarrow REPEAT)$
	\EndIf
	\item[]
	\Comment{Compute initial state for this subsystem}
	\If {$op = a$} $S \gets R[a]$	
	\ElsIf {$op = !a$} $S \gets R[!a]$
	\ElsIf {$op \in \{\vee,\wedge\}$} $S \gets S^1 \cup S^2 \cup R[\phi]B$
	\ElsIf {$op = U$} $S \gets S^1 \cup S^2 \cup R[\phi]A$
	\ElsIf {$op \in \{\Box,\Diamond\}$} $S \gets S^1 \cup R[\phi]$
	\ElsIf {$op \in \{X,W\}$} $S \gets R[\phi]$
	\EndIf
	\item[]
	\Comment{Compute and add reactivation rules for main operator}
	\If {$op \in \{\vee,\wedge\}$} $R_R \gets R_R \cup ([\phi]?Z \rightarrow R[\phi]?Z)$, for $Z \in {L,R,B}$
	\ElsIf {$op = U$} $R_R \gets R_R \cup ([\phi]?Z \rightarrow R[\phi]?Z,S^1,S^2)$, for $Z \in {A,B,L,R}$
	\ElsIf {$op \in \{\Box,\Diamond\}$}  $R_R \gets R_R \cup ([\phi]? \rightarrow R[\phi],S^1)$
	\ElsIf {$op \in \{X,W\}$} $R_R \gets R_R \cup ([\phi]? \rightarrow R[\phi]M,S^1) \cup ([\phi]?M \rightarrow R[\phi]M)$
	\EndIf
	\item[]
	\Comment{Return computed system}
	\State \Return $\langle R_E, R_R, S \rangle$
 	\EndFunction
\end{algorithmic}
}%end script
\end{algorithm}
%\vspace*{-5mm}

The complete set of evaluation tables is reported in Fig.~\ref{fig:tblz}, while the generation of evaluation and reactivation rules is summarised in Algorithm~\ref{alg:init}. The algorithm parses $\phi$ in a tree and visits the parsing tree in post-order. The system is built incrementally, starting from the system(s) returned by the recursive call(s). If $\phi$ is an observation (or its negation), an initial system is created, including two evaluation rules (as the observation may or may not occur), no reactivation rules and the single $R[\phi]$ as initial state. If $\phi$ is a conjunction or disjunction, the two systems of the subformulae are merged, and the conjunction/disjunction evaluation rules, reactivation rule and initial activation are added. The computations are the same if the main operator is $U$, but the reactivation rule will have to reactivate the monitoring of the two subformulae; in particular, $U_A$ denotes the standard {\em until} operator, while $U_B$ is the particular case where the $\psi$ failed and the {\em until} operator cannot be trivially satisfied anymore. Formulae with $X$ or $W$ as main operator go through two phases: first, the formula is evaluated to undecided, as the truth value can't be computed until the next cell is accessed. Special evaluation rules force the truth value to false (for $X$) or true (for $W$) if no next cell exists. Then, at the next iteration, the reactivation rule triggers the subformula: this means that if $X\phi$ is monitored in cell $i$, $\phi$ is monitored in cell $i+1$. $\phi$ is then monitored independently, and the $X\phi$ (or $W\phi$) rule enters a 'monitoring state' (suffix M in the table), simply mirroring $\phi$ truth value and self-reactivating. The evaluation of $\Box \phi$ is false (undecided) when $\phi$ is false (undecided); it is also undecided when $\phi$ holds (as $\Box \phi$ can never be true before the end of the trace), but the K suffix indicates when, at the end of the trace, an undecided $\Box$ can be evaluated to true. Finally, $\Diamond \phi$ constantly reactivates itself and its subformula $\phi$, unless $\phi$ is verified at runtime (causing $\Diamond \phi$ to hold), the trace ends ($\Diamond \phi$ fails).\\

RuleRunner generates several rules for each operator, but this number is constant, as it corresponds to the size of evaluation tables plus special rules (like the SUCCESS one). The number of rules corresponding to $\phi \vee \psi$, for instance, does not depend in any way on the nature of $\phi$ or $\psi$, as only the final truth evaluation of the two subformulae is taken into account. The preprocessing phase creates the parse tree of the property to encode and adds a constant number of rules for each node (subformula), and therefore the size of the rule set is linear w.r.t. the structure of the encoded formula $\phi$. The obtained rule set does not change at runtime nor when monitoring new traces. 

\subsection{Verification through Forward Chaining}

A RuleRunner rule system $RR_{\phi}$ encodes a FLTL formula $\phi$ in a set of rules. $RR_{\phi}$ can be used to check whether a given trace $t$ verifies or falsifies $\phi$. Given a set of Horn clauses in rule form $R$ and a set of atoms $A$, let the $FC(\cdot)$ (Forward Chaining) function be:
$$FC(R,A) = \{\beta \mid (A_i \rightarrow \beta) \in R \ \wedge \ A_i \subseteq A\}$$
Algorithm~\ref{alg:monitoringRR} describes how RuleRunner exploits forward chaining to perform a runtime verification task.\\

\begin{algorithm}
	{\scriptsize
	\caption{Runtime Verification using $RR_{\phi}$}
	\label{alg:monitoringRR}
\begin{algorithmic}[1]
	\Function{NN-monitor}{$\phi$,trace t} 
	\State \mbox{Create $RR_{\phi}=\langle R_R,R_E,S\rangle$ encoding $\phi$ (Algorithm~\ref{alg:init})}
	\While{new observations exist in t}
	\State $S' \gets S \cap obs$
	\While{$S \not= S'$}
	\State $S = S'$
	\State $S' \gets S \cap FC(S,R_E)$
	\EndWhile
	\If {$S$ contains SUCCESS (resp.FAILURE)}
    \State \Return return SUCCESS (resp.FAILURE)
	\EndIf
	\State $S \gets FC(S,R_R)$
	\EndWhile
	\EndFunction

\end{algorithmic}
}
\end{algorithm}

At the beginning, the rule system $RR_{\phi}$ is created. The monitoring loop iterates until $SUCCESS$ or $FAILURE$ is computed, and the FLTL semantics guarantees this happen in the last cell, if reached. At the beginning of each iteration (corresponding to the monitoring of a cell), the initial state $S$ contains a set of rule names corresponding to the subformulae to be checked in that cell. The observations of that cell are then added to the state of the system, and the state is incrementally expanded by means of forward chaining using the evaluation rules (line 7). This corresponds to computing the truth values of all subformulae of $\phi$ in a bottom-up way, from simple atoms to $\phi$ itself. If the monitoring did not compute a final verdict ($SUCCESS$/$FAILURE$), the state for the next cell is computed with a single application of $FC(\cdot)$ using the reactivation rules (line 12). Note that in this case the state is not expanded, as only the output of the forward chaining is stored ($S' \gets S \cap FC(S,R_E)$ vs $S \gets FC(S,R_R)$). This is used to {\em flush} all the previous truth evaluation, which are to be computed from scratch in the new cell.\\

During the runtime verification, for each cell, the $FC(\cdot)$ function is applied to the initial observations until the transitive closure of all evaluation rules is computed. The number of applications depends linearly on the encoded formula $\phi$: at each iteration the truth values of new subformulae are added, proceeding bottom-up from atoms to $\phi$. For instance, if $\phi = a \vee \Diamond b$, the first iteration would compute the truth values for $a$ and $b$, the second would add to the state the truth evaluation for $\Diamond b$, and finally the third one would compute the truth value of $\phi$ in the current cell. Therefore, for each cell the number of iterations of $FC(\cdot)$ is linear w.r.t. the structure of $\phi$. Each application of $FC(\cdot)$ depends on the number of rules and is again linear w.r.t. the structure of $\phi$, as stated in the previous subsection. This would suggest a quadratic complexity. However, in our implementation, (for each cell of the trace) the system goes through all rules exactly once. This is obtained by the post-order visit of the parsing tree, as shown in Algorithm~\ref{alg:init}, assuring pre-emption for rules evaluating simpler formulae. Therefore, the complexity of the system is inherently linear. This is not in contrast with known exponential lower bounds for the temporal logic validity problem, as RuleRunner deals with the satisfiability of a property on a trace, thus tackling a different problem from the validity one (this distinction is also mentioned in~\cite{trover}). \\

As an example, consider the formula $\phi = a \vee \Diamond b$ and the trace $t = [c - a - b,d - b,END]$ (dashes separate cells and commas separate observations in the same cell). Intuitively, $\phi$ means {\em either $a$ now or $b$ sometimes in the future.} If monitoring $\phi$ over $t$, $a$ fails straight from the beginning, while $b$ is sought until the third cell, when it is observed. Thus the monitoring yields a success even before the end of the trace. \\

In RuleRunner, for first, the formula $\phi$ is parsed into a tree, with $\vee$ as root and $a,b$ as leaves. Then, starting from the leaves, evaluation and reactivation rules for each node are added to the (initially empty) rule system. In our example, (part of) the rule system obtained from $\phi$, namely $RR_{(a \vee \Diamond b)}$, and its behaviour over $t$ are the following:\\

%\vspace*{-5mm}
\begin{minipage}[t]{0.5\textwidth}
{\scriptsize
\noindent
EVALUATION RULES
	\begin{itemize}
		%\item $R[a]$, $a$ is observed $\rightarrow$ $[a]{T}$
		\item $R[a],$ $a$ is not observed $\rightarrow$ $[a]{F}$
		\item $R[b]$, $b$ is observed $\rightarrow$ $[b]{T}$
		\item $R[b],$ $b$ is not observed $\rightarrow$ $[b]{F}$
		\item $R[\Diamond b]$, $[b]{T}$  $\rightarrow$ $[\Diamond b]{T}$
		%\item $R[\Diamond b]$, $[b]{?}$  $\rightarrow$ $[\Diamond b]{?}$
		\item $R[\Diamond b]$, $[b]{F}$  $\rightarrow$ $[\Diamond b]{?}$
		%\item $[\Diamond b]?$, $[END]$ $\rightarrow$ $[\Diamond b]{F}$
		%\item $R[a\vee\Diamond b]B$, $[a]{T}$, $[\Diamond b]{T}$ $\rightarrow$ $[a\vee\Diamond b]{T}$
		%\item $R[a\vee\Diamond b]B$, $[a]{T}$, $[\Diamond b]{?}$ $\rightarrow$ $[a\vee\Diamond b]{T}$
		%\item $R[a\vee\Diamond b]B$, $[a]{T}$, $[\Diamond b]{F}$ $\rightarrow$ $[a\vee\Diamond b]{T}$
		%\item $R[a\vee\Diamond b]B$, $[a]{?}$, $[\Diamond b]{T}$ $\rightarrow$ $[a\vee\Diamond b]{T}$
		%\item $R[a\vee\Diamond b]B$, $[a]{?}$, $[\Diamond b]{?}$ $\rightarrow$ $[a\vee\Diamond b]{?B}$
		%\item $R[a\vee\Diamond b]B$, $[a]{?}$, $[\Diamond b]{F}$ $\rightarrow$ $[a\vee\Diamond b]{?L}$
		%\item $R[a\vee\Diamond b]B$, $[a]{F}$, $[\Diamond b]{T}$ $\rightarrow$ $[a\vee\Diamond b]{T}$
		\item $R[a\vee\Diamond b]B$, $[a]{F}$, $[\Diamond b]{?}$ $\rightarrow$ $[a\vee\Diamond b]{?R}$
		%\item $R[a\vee\Diamond b]B$, $[a]{F}$, $[\Diamond b]{F}$ $\rightarrow$ $[a\vee\Diamond b]{F}$
		%\item $R[a\vee\Diamond b]L$, $[a]{T}$ $\rightarrow$ $[a\vee\Diamond b]{T}$
		%\item $R[a\vee\Diamond b]L$, $[a]{?}$ $\rightarrow$ $[a\vee\Diamond b]{?L}$
		%\item $R[a\vee\Diamond b]L$, $[a]{F}$ $\rightarrow$ $[a\vee\Diamond b]{F}$
		\item $R[a\vee\Diamond b]R$, $[\Diamond b]{T}$ $\rightarrow$ $[a\vee\Diamond b]{T}$
		\item $R[a\vee\Diamond b]R$, $[\Diamond b]{?}$ $\rightarrow$ $[a\vee\Diamond b]{?R}$
		%\item $R[a\vee\Diamond b]R$, $[\Diamond b]{F}$ $\rightarrow$ $[a\vee\Diamond b]{F}$
		\item $[a\vee\Diamond b]T$ $\rightarrow$ $SUCCESS$ \\
		%\item $[a\vee\Diamond b]F$ $\rightarrow$ $FAILURE$
	\end{itemize}
}

{\scriptsize
\noindent
REACTIVATION RULES	
		\begin{itemize}
%			\item $[\Diamond b]?, !END \rightarrow R[b],R[\Diamond b]$
%			\item $[\Diamond b]?, END \rightarrow R[\Diamond b]$
			\item $[\Diamond b]? \rightarrow R[b], R[\Diamond b]$
%			\item $[a \vee \Diamond b]?B \rightarrow R[a \vee \Diamond b]B$
%			\item $[a \vee \Diamond b]?L \rightarrow R[a \vee \Diamond b]L$
			\item $[a \vee \Diamond b]?R \rightarrow R[a \vee \Diamond b]R$ \\
		\end{itemize}
}

\end{minipage}
\begin{minipage}[t]{0.5\textwidth}

{\scriptsize
\noindent
INITIAL STATE	
\begin{itemize}
\item $R[a], R[b], R[\Diamond b], R[a \vee \Diamond b]B$ \\
\end{itemize} 
}

{\scriptsize
\noindent
EVOLUTION OVER $[c - a - b,d - b,END]$ \\

\begin{tabular}{ |r||l|}
	\hline
state & $R[a], R[b], R[\Diamond b], R[a \vee \Diamond b]B$ \\ \hline
+ obs & $R[a], R[b], R[\Diamond b], R[a \vee \Diamond b]B, c$ \\ \hline
eval & $[a]F, [b]F, [\Diamond b]?, [a \vee \Diamond b]?R$ \\ \hline
react & $R[b], R[\Diamond b], R[a \vee \Diamond b]R$ \\ \hline
\multicolumn{2}{c}{} \vspace*{-3mm}\\
	\hline
state & $R[b], R[\Diamond b], R[a \vee \Diamond b]R$ \\ \hline
+ obs & $R[b], R[\Diamond b], R[a \vee \Diamond b]R,a$ \\ \hline
eval & $[b]F, [\Diamond b]?, [a \vee \Diamond b]?R$ \\ \hline
react & $R[b], R[\Diamond b], R[a \vee \Diamond b]R$ \\ \hline
\multicolumn{2}{c}{} \vspace*{-3mm}\\
	\hline
state & $R[b], R[\Diamond b], R[a \vee \Diamond b]R$ \\ \hline
+ obs & $R[b], R[\Diamond b], R[a \vee \Diamond b]R,b,d$ \\ \hline
eval & $[b]T, [\Diamond b]T, [a \vee \Diamond b]T, SUCCESS$ \\ \hline
STOP & PROPERTY SATISFIED \\ \hline
\end{tabular}
}

\end{minipage}

\vspace*{5mm}
The behaviour of the runtime monitor is the following: 
\begin{itemize}
	\item At the beginning, the system monitors $a$,$b$,$\Diamond b$ and $a \vee \Diamond b$ (initial state = $R[a], R[b], R[\Diamond b],$ $ R[a \vee \Diamond b]B$). The $-B$ in $R[a \vee \Diamond b]B$ means that both disjuncts are being monitored.
	\item In the first cell, $c$ is observed and added to the state $S$. Using the evaluation rules, new truth values are computed: $a$ is false, $b$ is false, $\Diamond b$ is undecided. The global formula is undecided, but since the trace continues the monitoring goes on. The $-R$ in $R[a \vee \Diamond b]R$ means that only the right disjunct is monitored: the system dropped $a$, since it could only be satisfied in the first cell.
	\item In the second cell, $a$ is observed but ignored (the rules for its monitoring are not activated); since $b$ is false again, $\Diamond b$ and $a \vee \Diamond b$ are still undecided.
	\item In the third cell, $d$ is ignored but observing $b$ satisfies, in cascade, $b$, $\Diamond b$ and $a \vee \Diamond b$. The monitoring stops, signalling a success. The rest of the trace is ignored.\\
\end{itemize}

%Most of the runtime verification approaches, like RuleR \cite{ruler}, are based on a (potentially exponential) tree-like structure of alternative hypotheses, while RuleRunner is rooted in maintaining a single state composed of the unique truth value of every subformula of the encoded property. Moreover, the {\em distributed} nature of a RuleRunner state and the local nature of rules (inferring the truth value of a formula from the truth values of its subformulae only) allows different rules to be applied in parallel.
%Moreover, the fixed number of elements (observations, rule names, truth evaluations) and the implicit representation of time (by means of rule reactivation) allow to build a network with fixed structure and features: this will be explained in the next subsections.\\

\subsection{Semantics}
RuleRunner implements the FLTL \cite{fltl} semantics; however, there are two main differences in the approach. Firstly, FLTL is based on rewriting judgements, and it has no constraints over the accessed cells, while RuleRunner is forced to complete the evaluation on a cell before accessing the next one. Secondly, FLTL proceeds top-down, decomposing the property and then verifying the observations; RuleRunner propagates truth values bottom up, from observations to the property. In order to show the correspondence between the two formalisms, we introduce the map function:
\begin{center} $map : $ Property $\rightarrow$ FLTL judgement \end{center}
The $map$ function translates the state of a RuleRunner system into a FLTL judgement, analysing the state of the RuleRunner system monitoring $\phi$. Since $\Box$ and $\Diamond$ are derivate operators and they don't belong to FLTL specifications, we omit them from the discussion in this section.\\

\begin{algorithmic}
	\Function{map}{$\phi$, State,index}
\If {$SUCCESS \in $ State} \Return $\top$
\ElsIf {$FAILURE \in $ State} \Return $\perp$
\ElsIf {$[\phi]T \in $ State} \Return $\top$
\ElsIf {$[\phi]F \in $ State} \Return $\perp$
\ElsIf {$[\phi]?S \in $ State} $aux \gets S$
\Else \ find {$R[\phi]S \in $ State}; $aux \gets S$
\EndIf
\If {$\phi = a$}
	\State \Return $[u,index \models a]_{F}$
\ElsIf {$\phi =  \ !a$}
	\State \Return $[u,index \models \neg a]_{F}$
\ElsIf {$\phi = \psi^1 .. \psi^2 \ and \  aux = L$}
\State \Return $map(\psi^1)$
\ElsIf {$\phi = \psi^1 .. \psi^2 \ and \ aux = R$}
\State \Return $map(\psi^2)$	
\ElsIf {$\phi = \psi^1 \vee \psi^2 \ and \ aux = B$}
	\State \Return $map(\psi^1) \sqcup map(\psi^2)$
\ElsIf {$\phi = \psi^1 \wedge \psi^2 \ and \ aux = B$}
	\State \Return $map(\psi^1) \sqcap map(\psi^2)$
\ElsIf {$\phi = \psi^1 U \psi^2 \ and \ aux = A$}
	\State \Return $map(\psi^2) \sqcup (map(\psi^1) \sqcap (map(X(\psi^1 U \psi^2))))$
\ElsIf {$\phi = \psi^1 U \psi^2 \ and \ aux = B$}
	\State \Return $map(\psi^2) \sqcap (map(X(\psi^1 U \psi^2)))$	next
\ElsIf {$\phi = X\psi \ and \ aux \not=M$}
	\State \Return $[u,index \models X\psi]_{F}$
\ElsIf {$\phi = W\psi \ and \ aux \not=M$}
	\State \Return $[u,index \models \bar{X}\psi]_{F}$
\ElsIf {$(\phi = X\psi \ or \ \phi = W\psi) \ and \ aux = M$}
	\State \Return $map(\psi)$	
\EndIf
 \EndFunction
\end{algorithmic}
The following table reports a simple example of an evolution of a RuleRunner step and the corresponding value computed by $map$. Let the property be $a \vee X b$ and the trace be $u = [b-b]$. The index is incremented when the reactivation rules are fired. 
\begin{table}
\begin{center}
\begin{tabular}{ |l|l|}
\hline 
{\bf State} & {\boldmath$map(a \vee X b)$} \\ \hline
{$R[a],R[Xb],R[a \vee X b]B$} & {$[u,0 \models a]_F \sqcup [u,0 \models Xb]_F$} \\ \hline
{$R[a],R[Xb],R[a \vee X b]B,b$} & {$[u,0 \models a]_F \sqcup [u,0 \models Xb]_F$} \\ \hline
{$R[a],R[Xb],R[a \vee X b]B,b,[a]F$} & {$\perp \sqcup [u,0 \models Xb]_F$} \\ \hline
{$R[a],R[Xb],R[a \vee X b]B,b,[a]F, [b]?M$} & {$\perp \sqcup [u,0 \models Xb]_F$} \\ \hline
{$R[a],R[Xb],R[a \vee X b]B,b,[a]F, [b]?M, [a \vee Xb]?R$} & {$[u,0 \models Xb]_F$} \\ \hline
{$R[b],R[Xb]M,R[a \vee X b]R$} & {$[u,1 \models b]_F $} \\ \hline
{$R[b],R[Xb]M,R[a \vee X b]R, b$} & {$[u,1 \models b]_F $} \\ \hline
{$R[b],R[Xb]M,R[a \vee X b]R, b, [b]T$} & {$\top $} \\ \hline
{$R[b],R[Xb]M,R[a \vee X b]R, b, [b]T, [Xb]T$} & {$\top $} \\ \hline
{$R[b],R[Xb]M,R[a \vee X b]R, b, [b]T, [Xb]T, [a \vee Xb]T$} & {$\top $} \\ \hline
{$SUCCESS$} & {$\top $} \\ \hline
\end{tabular}
\vspace*{2mm}
\caption{The $map$ function}
\end{center}
\end{table}

\begin{theorem} For any well-formed FLTL formula $\phi$ over a set of observations, and for every finite trace $u$, for every intermediate state $s_i$ in RuleRunner's evolution over $u$ there exist a valid rewriting $r_j$ of $[u,0 \models \phi]_F$ such that $map(\phi) = r_j$. In other words, RuleRunner's state can always be mapped onto an FLTL judgement over $\phi$.
\end{theorem}
\begin{proof} The proof proceeds by induction on $\phi$:
	\begin{itemize}
		\item {\boldmath$\phi = a$}\\ If the formula is a simple observation, then the initial state is $R[a]$, and $map(R[a]) = [u,0 \models a]_{F}$. Adding observation to the state does not change the resulting FLTL judgement. If $a$ is observed, RuleRunner will add $[a]T$ to the state, and this will be mapped to $\top$. If $a$ is not observed, RuleRunner will add $[a]F$ to the state, and this will be mapped to $\perp$. So for this simple case, the evolution of RuleRunner's state corresponds either to the rewriting $[u,0 \models a]_{F} = \top$ (if $a$ is observed) or to the rewriting $[u,0 \models a]_{F} = \perp$ (if $a$ is not observed).
		\item {\boldmath$\phi = !a$}\\
		This case is analogous tho the previous one, with opposite verdicts.
		\item {\boldmath $\phi = \psi^1 \vee \psi^2$} \\		
		By inductive hypothesis, a RuleRunner system monitoring $\psi^1$ always corresponds to a rewriting of $[u,i\models \psi^1]$. The same holds for $\psi^2$. Let $\langle R_R^i, R_E^i, S^i \rangle$ be RuleRunner system monitoring the subformula $\psi^1$, with $i \in  \{ 1,2 \}$. A RuleRunner system encoding $\phi$ includes $R^1$ and $R^2$ rules and specific rules for $\psi^1 \vee \psi^2$ given the truth values of $\psi^1$ and $\psi^2$. The initial state is therefore $R[\psi^1 \vee \psi^2] \cup S^1 \cup S^2$, and this is mapped to $map(S^1) \sqcup map(S^2)$. By inductive hypothesis, this is a valid FLTL judgement. In each iteration, as long as the truth value of $\psi^1 \vee \psi^2$ is not computed, the state is mapped on $map(S^1) \sqcup map(S^2)$. When the propagation of truth values reaches $\psi^1 \vee \psi^2$, the assigned truth value mirrors the evaluation table for the disjunction. If either $\psi^1$ or $\psi^2$ is true, then $\phi$ is true, and $map(\phi) = \top$. This corresponds to the valid rewriting $map(S^1) \sqcup map(S^2) = \top$, given that we are considering the case in which there is a true $\psi^i$: $[\psi^i]T$ belongs to the state and $map(\psi^1) = \top$. The false-false case is analogous. In the $?_B$ case, the mapping is preserved, and this is justified by the fact that both $\psi^1$ and $\psi^2$ are undecided in the current cell, therefore $map(\psi^i) \not = \top,\perp$, therefore $map(\psi^1) \sqcup map(\psi^2)$ could not be simplified. In the $?_L$ case, we have that $[\psi^2]F$, therefore $map(\psi^2) = \perp$. The FLTL rewriting is $map(\psi^1) \sqcup map(\psi^2) = map(\psi^1)$, and this is a valid rewriting since $map(\psi^1) \sqcup map(\psi^2) = map(\psi^1) \sqcup \perp = map(\psi^1)$. The $?_R$ case is symmetrical.
	\item {\boldmath $\phi = \psi^1 \wedge \psi^2$} \\	
	Same as above, with the evaluation table for conjunction on the RuleRunner side and the $\sqcap$ operator on the FLTL judgement side.
	\item {\boldmath $\phi = X\psi$} \\
	A RuleRunner system encoding $X\phi$ has initial state $R[X\phi]$, which is mapped on $[u,0 \models X\psi]_F$. Then, if the current cell is the last one, $R[X\phi]$ evaluates to $[X\phi]F$, and the corresponding FLTL judgement is $\perp$. If another cell exists, $R[X\phi]$ evaluates to $[X\phi]?$ (with the same mapping). When the reactivation rules are triggered, $[X\phi]?$ is substituted by $R[X\psi]M,R[\psi]$. Over this state, $map(X\psi) = map (\psi)$, and the index is incremented since reactivation rules were fired. Therefore, the FLTL rewriting is $[u,i \models X\psi] = [u,i+1 \models \psi]$, and this is a valid rewriting.
	\item {\boldmath $\phi = W\psi$} \\
	This case is like the previous, but if the current cell is the last then $R[W\psi]$ evolves to $[W\psi]T$; the mapping is rewritten from $[u,i \models W\psi]$ to $\top$, and this is a valid rewriting if there is no next cell.
	\item {\boldmath $\phi = \psi^1 U \psi^2$} \\
	The initial RuleRunner system includes rules for $\psi^1$, $\psi^2$ and for the $U$ operator. As long as $R[\psi^1 U \psi^2]A$ is not evalued, $map(\psi^1 U \psi^2) = map(\psi^2) \sqcup (map(\psi^1) \sqcap (map(X(\psi^1 U \psi^2))))$, that is, the standard one-step unfolding of the 'until' operator as defined in FLTL. When a truth value for the global property is computed, there are several possibilities. The first one is that $\psi^2$ is true and $\psi^1 U \psi^2$ is immediately satisfied. RuleRunner adds $[\psi^1 U \psi^2]T$ to the state and $map(\phi) = \top$; this corresponds to the rewriting $map(\psi^2) \sqcup (map(\psi^1) \sqcap (map(X(\psi^1 U \psi^2)))) = \top \sqcup (map(\psi^1) \sqcap (map(X(\psi^1 U \psi^2)))) = \top$, which is a valid rewriting. The case for $[\psi^1]F$ and $[\psi^2]F$ is analogous. The $?_A$ case means that the evaluation for the until is undecided in the current trace, and is mapped on the standard one-step unfolding of the until operator in FLTL. The $?_B$ case implicitly encode the information that 'the until cannot be trivially satisfied anymore', and henceforth the FLTL mapping is $map(\psi^1) \sqcap (map(X(\psi^1 U \psi^2)))$. The cases for $?_L$ and $?_R$ have the exact meaning they had in the disjunction and conjunction cases. For instance, if $[\psi^1]F$ and $[\psi^2]?$, RuleRunner adds $[\psi^1 U \psi^2]?R$ to the state, and for the obtained state $map(\phi) = map(\psi^2)$. The sequence of FLTL rewriting is $map(\psi^2) \sqcup (map(\psi^1) \sqcap (map(X(\psi^1 U \psi^2)))) = map(\psi^2) \sqcup (\perp \sqcap (map(X(\psi^1 U \psi^2)))) = map(\psi^2) \sqcup \perp = map(\psi^2)$.
	\end{itemize}
\end{proof}
\begin{corollary}
	RuleRunner yields a FLTL verdict.
\end{corollary}
\begin{proof}
	RuleRunner is always in a state that can be mapped on a valid FLTL judgement; therefore, when a binary truth evaluation for the encoded formula is given, this is mapped on the correct binary evaluation in FLTL. But since for such trivial case the $map$ function corresponds to an identity, the RuleRunner evaluation is a valid FLTL judgement. The fact that RuleRunner yields a binary verdict is guaranteed provided that the analysed trace is finite, thanks to end-of-trace rules.
\end{proof}

\section{Experiments}
In order to test the scalability of our system, we tested our prototype against several properties and traces: in this paragraph we report a simple set of experiments and results. These tests involve three FLTL formulae, respectively $\phi^1 = \Diamond a$, $\phi^2 = \Box((a \vee b)\vee(c \vee d))$ and $\phi^3 = \Diamond((a\wedge Xb)\vee(c \wedge Wd))$. We encoded $\phi^{1,2,3}$ in three RuleRunner systems and used them to monitor traces randomly generated using the Latin alphabet as set of observations. Note that each monitoring process can terminate before the end of the trace (e.g. trivially, if $a$ is observed while monitoring $\Diamond a$); we measured the time required to actually monitor a given number of cells. 

\begin{figure}[ht!]
	\centering
	%\hspace*{-9mm}
	\includegraphics[scale=0.5]{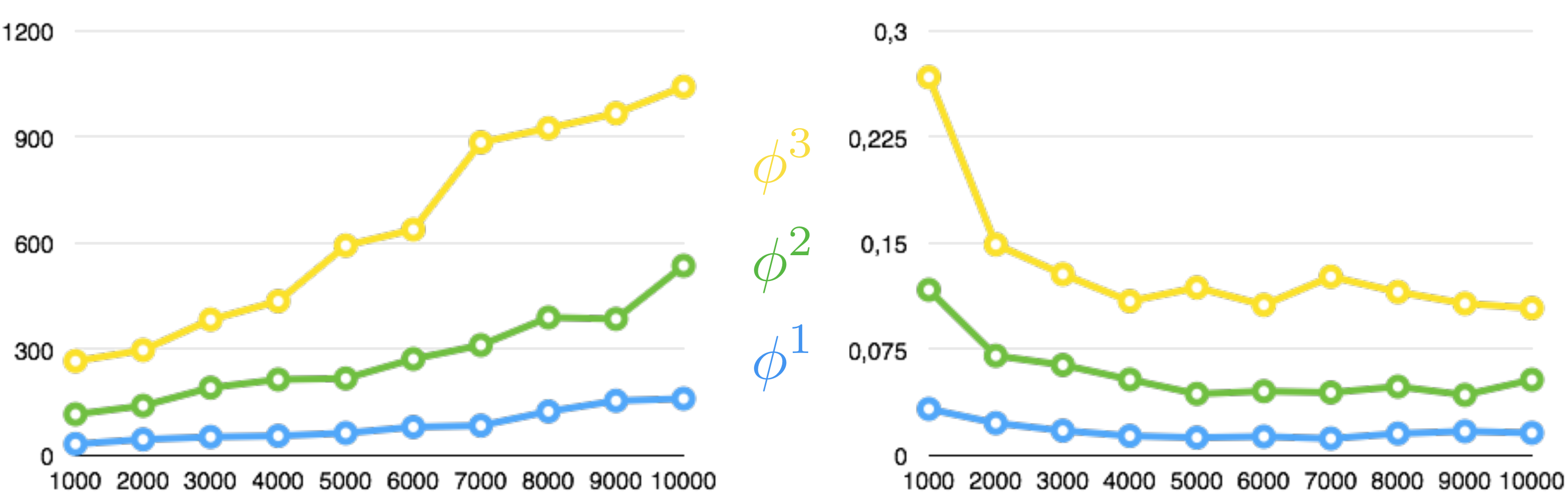}
	\vspace*{-3mm}
	\caption{Absolute (left) and averaged (right) performance of monitors encoding $\phi^{1,2,3}$}
	\label{fig:xp}
\end{figure}

Figure~\ref{fig:xp} shows the time required, for the three rule systems, to monitor an increasing number of cells. In both subfigures, the {\em x-axis} represent the number of monitored cells and the {\em y-axis} a time measurement in milliseconds ($ms$). The three curves, in both subfigures, correspond to the three monitors for $\phi^1 = \Diamond a$, $\phi^2 = \Box((a \vee b)\vee(c \vee d))$ and $\phi^3 = \Diamond((a\wedge Xb)\vee(c \wedge Wd))$. Figure~\ref{fig:xp}(A) reports total times and Figure~\ref{fig:xp}(B) reports average monitoring time per cell. The trends show how the monitoring time scales w.r.t. the number of the traces; the decreasing of average times in the curves of Figure~\ref{fig:xp}(B) is due to the overhead to compute the rule system becoming less relevant when averaging with a larger number of cells.\\
These experiment can be replicated with the tool available at {\href{http://www.di.unito.it/~perotti/RuleRunner.jnlp}{www.di.unito.it/$\sim$perotti/RuleRunner.jnlp}}.

%%%%%%%%%%%%%%%%%%%%%%%%%%%%%%%%%%%%%%%%%%%%%%%%%%%%%%%%%%%%%%%%%%%%%%%%%%%%%%%%%%%%
%%%%%%%%%%%%%%%%%%%%%%%%%%%%%%%%%%%%%%%%%%%%%%%%%%%%%%%%%%%%%%%%%%%%%%%%%%%%%%%%%%%%
%%%%%%%%%%%%%%%%%%%%%%%%%%%%%%%%%%%%%%%%%%%%%%%%%%%%%%%%%%%%%%%%%%%%%%%%%%%%%%%%%%%%

\section{Conclusions and Future Work}
In this paper we present RuleRunner, a rule-based runtime verification system that exploits Horn clauses in implication form and forward chaining to perform a monitoring task. RuleRunner is a module in a wider framework that includes the encoding of the rule system in a neural network, the exploitation of GPUs to improve monitoring performances (as computation in neural networks boils down to matrix-based operations) and the adoption of machine learning algorithms to adapt the encoded property to the observed trace. Our final goal is the development of a system for scalable and parallel monitoring and capable to provide a description of patterns that falsified the prescribed temporal property. The applications of this frameworks spans from multi-agent systems (where a system designer may want to use an agent's unscripted solution to a problem as a benchmark for all other agents~\cite{nmas}) to security (where a security manager may want to correct some false positives when monitoring security properties~\cite{falsepos}).\\
Concerning RuleRunner, a future direction of work is to create rule systems for other finite-trace semantics. For instance, we conjecture that removing all rules with $END$ would be a valid starting point for the development of a rule system for LTL3; the rule systems for FLTL and LTL3 will then be used to build a rule system for RVLTL. A second direction of future work will be to modify RuleRunner in such a way to use external forward chaining tools for the monitoring (as we use our own inference engine), such as the Constraint Handling Rules extension included in several Prolog implementations.\\

\section{Bibliography}
\bibliographystyle{eptcs}
\bibliography{hvcs14}

\begin{thebibliography}{10}
\providecommand{\bibitemdeclare}[2]{}
\providecommand{\surnamestart}{}
\providecommand{\surnameend}{}
\providecommand{\urlprefix}{Available at }
\providecommand{\url}[1]{\texttt{#1}}
\providecommand{\href}[2]{\texttt{#2}}
\providecommand{\urlalt}[2]{\href{#1}{#2}}
\providecommand{\doi}[1]{doi:\urlalt{http://dx.doi.org/#1}{#1}}
\providecommand{\bibinfo}[2]{#2}

\bibitemdeclare{inproceedings}{AAlst}
\bibitem{AAlst}
\bibinfo{author}{Wil M.~P. \surnamestart van~der Aalst~et.al.\surnameend}:
  \emph{\bibinfo{title}{Process Mining Manifesto}}.
\newblock In: {\sl \bibinfo{booktitle}{Procs of Business Process Management
  Workshops 2011}}, pp. \bibinfo{pages}{169--194}.
\newblock \urlprefix\url{http://dx.doi.org/10.1007/978-3-642-28108-2_19}.

\bibitemdeclare{article}{barr}
\bibitem{barr}
\bibinfo{author}{Howard \surnamestart Barringer\surnameend},
  \bibinfo{author}{David~E. \surnamestart Rydeheard\surnameend} \&
  \bibinfo{author}{Klaus \surnamestart Havelund\surnameend}
  (\bibinfo{year}{2010}): \emph{\bibinfo{title}{Rule Systems for Run-time
  Monitoring: from Eagle to RuleR}}.
\newblock {\sl \bibinfo{journal}{Journal of Logic and Computation}}
  \bibinfo{volume}{volume 20}, pp. \bibinfo{pages}{pages 675--706}.
\newblock \urlprefix\url{http://dx.doi.org/10.1093/logcom/exn076}.

\bibitemdeclare{inproceedings}{rvltl}
\bibitem{rvltl}
\bibinfo{author}{Andreas \surnamestart Bauer\surnameend},
  \bibinfo{author}{Martin \surnamestart Leucker\surnameend} \&
  \bibinfo{author}{Christian \surnamestart Schallhart\surnameend}:
  \emph{\bibinfo{title}{The Good, the Bad, and the Ugly, But How Ugly Is
  Ugly?}}
\newblock In: {\sl \bibinfo{booktitle}{Procs. of Runtime Verification 2007}},
  pp. \bibinfo{pages}{126--138}.
\newblock \urlprefix\url{http://dx.doi.org/10.1007/978-3-540-77395-5_11}.

\bibitemdeclare{inproceedings}{ltl3}
\bibitem{ltl3}
\bibinfo{author}{Andreas \surnamestart Bauer\surnameend},
  \bibinfo{author}{Martin \surnamestart Leucker\surnameend} \&
  \bibinfo{author}{Christian \surnamestart Schallhart\surnameend}:
  \emph{\bibinfo{title}{Monitoring of Real-Time Properties}}.
\newblock In: {\sl \bibinfo{booktitle}{Procs. of Foundations of Software
  Technology and Theoretical Computer Science 2006}}, pp.
  \bibinfo{pages}{260--272}.
\newblock \urlprefix\url{http://dx.doi.org/10.1007/11944836_25}.

\bibitemdeclare{article}{falsepos}
\bibitem{falsepos}
\bibinfo{author}{David \surnamestart Breitgand\surnameend},
  \bibinfo{author}{Maayan \surnamestart Goldstein\surnameend} \&
  \bibinfo{author}{E.~H. \surnamestart Shehory\surnameend}
  (\bibinfo{year}{2011}): \emph{\bibinfo{title}{Efficient Control of False
  Negative and False Positive Errors with Separate Adaptive Thresholds}}.
\newblock {\sl \bibinfo{journal}{Network and Service Management, IEEE
  Transactions on}} \bibinfo{volume}{8}, pp. \bibinfo{pages}{128--140}.
\newblock \urlprefix\url{http://dx.doi.org/10.1109/TNSM.2011.020111.00055}.

\bibitemdeclare{inproceedings}{trover}
\bibitem{trover}
\bibinfo{author}{Doron \surnamestart Drusinsky\surnameend}:
  \emph{\bibinfo{title}{The Temporal Rover and the ATG Rover}}.
\newblock In: {\sl \bibinfo{booktitle}{Procs. of the International Workshop on
  SPIN Model Checking and Software Verification 2000}}, pp.
  \bibinfo{pages}{323--330}.
\newblock \urlprefix\url{http://dx.doi.org/10.1007/10722468_19}.

\bibitemdeclare{inproceedings}{ltl+}
\bibitem{ltl+}
\bibinfo{author}{Cindy \surnamestart Eisner\surnameend}, \bibinfo{author}{Dana
  \surnamestart Fisman\surnameend}, \bibinfo{author}{John \surnamestart
  Havlicek\surnameend}, \bibinfo{author}{Yoad \surnamestart Lustig\surnameend},
  \bibinfo{author}{Anthony \surnamestart McIsaac\surnameend} \&
  \bibinfo{author}{David~Van \surnamestart Campenhout\surnameend}:
  \emph{\bibinfo{title}{Reasoning with Temporal Logic on Truncated Paths}}.
\newblock In: {\sl \bibinfo{booktitle}{Procs. of Computer-Aided Verification
  2003}}, pp. \bibinfo{pages}{27--39}.
\newblock \urlprefix\url{http://dx.doi.org/10.1007/978-3-540-45069-6_3}.

\bibitemdeclare{article}{cilp}
\bibitem{cilp}
\bibinfo{author}{Artur~S. \surnamestart d'Avila Garcez\surnameend} \&
  \bibinfo{author}{Gerson \surnamestart Zaverucha\surnameend}
  (\bibinfo{year}{1999}): \emph{\bibinfo{title}{The Connectionist Inductive
  Learning and Logic Programming System}}.
\newblock {\sl \bibinfo{journal}{Applied Intelligence}} \bibinfo{volume}{volume
  11}, pp. \bibinfo{pages}{pages 59--77}.
\newblock \urlprefix\url{http://dx.doi.org/10.1023/A:1008328630915}.

\bibitemdeclare{article}{nmas}
\bibitem{nmas}
\bibinfo{author}{Christopher~D. \surnamestart Hollander\surnameend} \&
  \bibinfo{author}{Annie~S. \surnamestart Wu\surnameend}
  (\bibinfo{year}{2011}): \emph{\bibinfo{title}{The Current State of Normative
  Agent-Based Systems.}}
\newblock {\sl \bibinfo{journal}{J. Artificial Societies and Social
  Simulation}} \bibinfo{volume}{14}, pp. \bibinfo{pages}{47--62}.
\newblock \urlprefix\url{http://jasss.soc.surrey.ac.uk/14/2/6.html}.

\bibitemdeclare{article}{horn}
\bibitem{horn}
\bibinfo{author}{Alfred \surnamestart Horn\surnameend} (\bibinfo{year}{1951}):
  \emph{\bibinfo{title}{On Sentences Which are True of Direct Unions of
  Algebras}}.
\newblock {\sl \bibinfo{journal}{J. Symb. Log.}} \bibinfo{volume}{16}, pp.
  \bibinfo{pages}{14--21}.
\newblock \urlprefix\url{http://dx.doi.org/10.2307/2268661}.

\bibitemdeclare{article}{kripke}
\bibitem{kripke}
\bibinfo{author}{Saul~A. \surnamestart Kripke\surnameend}
  (\bibinfo{year}{1963}): \emph{\bibinfo{title}{Semantical Considerations on
  Modal Logic}}.
\newblock {\sl \bibinfo{journal}{Acta Philosophica Fennica}}
  \bibinfo{volume}{16}, pp. \bibinfo{pages}{83--94}.

\bibitemdeclare{inproceedings}{hornprod}
\bibitem{hornprod}
\bibinfo{author}{Jean-Louis \surnamestart Lassez\surnameend} \&
  \bibinfo{author}{Michael~J. \surnamestart Maher\surnameend}:
  \emph{\bibinfo{title}{The Denotational Semantics of Horn Clauses as a
  Production System.}}
\newblock In: {\sl \bibinfo{booktitle}{Procs. of the Association for the
  Advancement of Artificial Intelligence 1983}}, pp. \bibinfo{pages}{229--231}.
\newblock \urlprefix\url{http://www.aaai.org/Library/AAAI/1983/aaai83-006.php}.

\bibitemdeclare{article}{rv}
\bibitem{rv}
\bibinfo{author}{Martin \surnamestart Leucker\surnameend} \&
  \bibinfo{author}{Christian \surnamestart Schallhart\surnameend}
  (\bibinfo{year}{2009}): \emph{\bibinfo{title}{A brief account of runtime
  verification}}.
\newblock {\sl \bibinfo{journal}{Journal of Logic and Algebraic Programming}}
  \bibinfo{volume}{volume 78}, pp. \bibinfo{pages}{pages 293--303}.
\newblock \urlprefix\url{http://dx.doi.org/10.1016/j.jlap.2008.08.004}.

\bibitemdeclare{inproceedings}{fltl}
\bibitem{fltl}
\bibinfo{author}{Orna \surnamestart Lichtenstein\surnameend},
  \bibinfo{author}{Amir \surnamestart Pnueli\surnameend} \&
  \bibinfo{author}{Lenore~D. \surnamestart Zuck\surnameend}:
  \emph{\bibinfo{title}{The Glory of the Past}}.
\newblock In: {\sl \bibinfo{booktitle}{in Procs. of Logic of Programs 1985}},
  pp. \bibinfo{pages}{196--218}.
\newblock \urlprefix\url{http://dx.doi.org/10.1007/3-540-15648-8_16}.

\bibitemdeclare{book}{three}
\bibitem{three}
\bibinfo{author}{J.~\surnamestart Lukasiewicz\surnameend}
  (\bibinfo{year}{1920}): \emph{\bibinfo{title}{O logice trójwartosciowej (On
  Three-Valued Logic)}}.

\bibitemdeclare{inproceedings}{ltl}
\bibitem{ltl}
\bibinfo{author}{Amir \surnamestart Pnueli\surnameend}:
  \emph{\bibinfo{title}{The temporal logic of programs}}.
\newblock In: {\sl \bibinfo{booktitle}{Procs. of the Annual Symposium on
  Foundations of Computer Science 1977}}, pp. \bibinfo{pages}{46--57}.
\newblock
  \urlprefix\url{http://doi.ieeecomputersociety.org/10.1109/SFCS.1977.32}.

\bibitemdeclare{article}{kow}
\bibitem{kow}
\bibinfo{author}{M.~H. \surnamestart Van~Emden\surnameend} \&
  \bibinfo{author}{R.~A. \surnamestart Kowalski\surnameend}
  (\bibinfo{year}{1976}): \emph{\bibinfo{title}{The Semantics of Predicate
  Logic As a Programming Language}}.
\newblock {\sl \bibinfo{journal}{J. ACM}} \bibinfo{volume}{23}, pp.
  \bibinfo{pages}{733--742}.
\newblock \urlprefix\url{http://dx.doi.org/10.1145/321978.321991}.

\end{thebibliography}
\end{document}